\documentclass[reqno]{amsart}

\usepackage[usenames,dvipsnames]{color}
\usepackage[bookmarks,colorlinks,breaklinks]{hyperref}  
\usepackage{doi,url}

\definecolor{dullmagenta}{rgb}{0.4,0,0.4}   
\definecolor{darkblue}{rgb}{0,0,0.4}
\hypersetup{linkcolor=dullmagenta,citecolor=blue,filecolor=dullmagenta,urlcolor=darkblue} 

\usepackage{amssymb,yhmath}
\usepackage{enumerate}

\newcommand{\eq}[1]{\eqref{#1}}

\newtheorem{theorem}{Theorem}[section]

\newtheorem{lemma}[theorem]{Lemma}

\newtheorem{definition}[theorem]{Definition}
\newtheorem{remark}[theorem]{Remark}

\newtheorem*{comment}{Comment}

\numberwithin{equation}{section}

\DeclareMathOperator{\supp}{supp}
\DeclareMathOperator{\tr}{tr}
\DeclareMathOperator{\Ran}{Ran}
\DeclareMathOperator{\dist}{dist}

\DeclareMathOperator{\Rea}{Re}
\DeclareMathOperator{\Ima}{Im}

\newcommand{\pr}{\prime}

\newcommand\R{\mathbb R}
\newcommand\N{\mathbb N}

\newcommand\Z{\mathbb Z}

\renewcommand\P{\mathbb P}
\newcommand\E{\mathbb E}

\newcommand\Dd{\mathbb D}

\renewcommand\H{\mathcal{H}}
\renewcommand\L{\mathrm{L}}
 
\newcommand\D{\mathcal{D}}

\newcommand{\cJ}{\mathcal{J}}
\newcommand{\cD}{\mathcal{D}}

\newcommand{\cT}{\mathcal{T}}

\newcommand\e{\mathrm{e}}

\newcommand{\bom}{{\boldsymbol{\omega}}}

\newcommand\La{\Lambda}
\newcommand{\vphi}{\varphi}

\newcommand\Chi{\raisebox{.2ex}{$\chi$}}

\newcommand{\abs}[1]{\left\lvert #1 \right\rvert}
\newcommand{\norm}[1]{\left\lVert #1 \right\rVert}
\newcommand{\scal}[1]{\left\langle #1 \right\rangle}
\newcommand{\set}[1]{\left\{ #1 \right\}}
\newcommand{\pa}[1]{\left( #1 \right)}

\newcommand{\up}[1]{^{(#1)}}

\newcommand\beq{\begin{equation}}
\newcommand\eeq{\end{equation}}

\newcommand{\qtx}[1]{\quad\text{#1}\quad}

\begin{document}

\title[Unique continuation principle for spectral projections]
{Unique continuation principle for spectral projections of Schr\" odinger operators and optimal Wegner estimates for non-ergodic random Schr\" odinger operators}

\author{Abel Klein}
\address{University of California, Irvine;
Department of Mathematics;
Irvine, CA 92697-3875,  USA}
 \email{aklein@uci.edu}

\thanks{A.K. was  supported in part by the NSF under grant DMS-1001509.}


\begin{abstract} 
We prove a unique continuation principle for  spectral projections of  Schr\" odinger operators. We consider     a Schr\" odinger operator $H= -\Delta + V$ on  $\mathrm{L}^2(\R^d)$, and  let $H_\La$ denote its restriction  to    a finite box $\La$ with either Dirichlet or periodic boundary condition.    We prove unique continuation estimates   of the type  $\Chi_I (H_\La)  W \Chi_I (H_\La) \ge \kappa\,   \Chi_I (H_\La) $ with $\kappa >0$ for appropriate potentials $W\ge 0$ and intervals $I$.  As an application, we  obtain optimal Wegner estimates   at all energies for  a class of non-ergodic  random  Schr\" odinger operators with alloy-type random potentials (`crooked' Anderson Hamiltonians).  We also prove optimal Wegner estimates at the bottom of the spectrum  with  the expected dependence on the disorder  (the Wegner estimate improves as the disorder increases), a new result even for the usual (ergodic) Anderson Hamiltonian.    These estimates are applied  to prove localization at high disorder for Anderson Hamiltonians in a fixed interval at the bottom of the spectrum. 
 \end{abstract}

\maketitle


\section{introduction}  

Let $H= -\Delta + V$ be a  Schr\" odinger operator  on  $\mathrm{L}^2(\R^d)$.  Given a box (or cube)  $\La=\La_L (x_0)\subset \R^d$ with side of length  $L  $ and center $x_0 \in \R^d$,  let $H_\La= -\Delta_\La+ V_\La$ denote the  restriction of $H$   to    the  box $\La$ with either Dirichlet or periodic boundary condition:    $\Delta_\La$ is the Laplacian  with either Dirichlet or periodic boundary condition and $V_\La$ is the restriction of $V$ to $\La$.  (We will abuse the notation and simply write $V$ for $V_\La$, i.e., $H_\La= -\Delta_\La+ V$ on $\L^2(\La)$.)
By  a unique continuation principle for  spectral projections  (UCPSP)   we will mean an estimate
of the form
\beq\label{UCPSP}
\Chi_I (H_\La)  W \Chi_I (H_\La) \ge \kappa\,   \Chi_I (H_\La),
\eeq 
where $\Chi_I$ is the characteristic function of  an interval $I \subset \R $, $W\ge 0 $ is a  potential, and $\kappa >0$ is a constant.

If $V$ and $W$ are bounded $\Z^d$-periodic  potentials,  $W\ge 0$  with  $W>0$ on some open set,
 Combes, Hislop and Klopp  \cite[Section~4]{CHK1}, \cite[Theorem~2.1]{CHK2} proved a UCPSP for  $H_\La$  with  periodic boundary condition, for  boxes $\La=\La_L (x_0)\subset \R^d$ with $L \in \N$ and $x_0 \in \Z^d$ and  arbitrary bounded  intervals $I$,  with a constant $\kappa >0$ depending  on $d,I,V,W$ but not on the box  $\La$. Their proof uses  the unique continuation principle and Floquet theory.   Germinet and Klein \cite[Theorem~A.6]{GKloc} proved a modified version of this result, using Bourgain and Kenig's quantitative unique continuation principle \cite[Lemma~3.10]{BK} and Floquet theory, obtaining  control of the constant $\kappa$ in terms of the relevant parameters.  

Rojas-Molina and Veseli\' c recently proved ``scale-free unique continuation estimates" for Schr\" odinger operators \cite[Theorem~2.1]{RV} (see also \cite[Theorem~A.1.1]{R}).  They consider a   Schr\" odinger operator $ H= -\Delta + V$, where $V$ is only required to be bounded, and its   restrictions  $H_\La$ to    boxes  $\La$ with side $L \in  \N$ with either Dirichlet or periodic boundary condition.  They decompose the box $\La$ into unit boxes, and for each unit box pick a ball of (a fixed)  radius $\delta$ contained in the unit box, and let   $W$ be the potential given by the sum of the characteristic functions of those balls. Using a version of the quantitative unique continuation principle \cite[Theorem~3.1]{RV},  they prove that if $\psi$ is an eigenfunction of $H_\La$ with eigenvalue $E$ (more generally, if $\abs{\Delta \psi} \le \abs {(V-E)\psi}$), then
\beq\label{RMV}
\norm{W\psi}_2^2 \ge \kappa \norm{\psi}_2^2,
\eeq
where the constant $\kappa >0$ depends only on  $d,V,\delta,E$, and is locally bounded on $E$.  
Since \eq{RMV} is just the UCPSP \eq{UCPSP} when $I=\set{E}$, this raises the question  of  the validity of a UCPSP  in this setting, posed   as an open question by Rojas-Molina and Veseli\' c   \cite{RV}.

In this article we prove   a UCPSP for Schr\" odinger operators (Theorem~\ref{thmUCPSP}), giving an affirmative answer to the  open question in  \cite{RV}.  The proof is based on  the quantitative unique continuation principle  derived by Bourgain and Klein  \cite[Theorem~3.2]{BKl},  restated here as Theorem~\ref{thmucp}.   This version of the quantitative unique continuation principle,   as the original result of Bourgain and Kenig  \cite[Lemma~3.10]{BK}  and the version of   Germinet and Klein  \cite[Theorem~A.1]{GKloc}, allows for   approximate solutions of the stationary Schr\" odinger equation.  (\cite[Theorem~3.1]{RV} requires   $\abs{\Delta \psi} \le \abs {V\psi}$.)  Theorem~\ref{thmucp}  can be applied not only to eigenfunctions of a Schr\" odinger operator $H$, but also to approximate eigenfunctions, i.e., 
arbitrary $\psi \in \Ran \Chi_{[E-\gamma ,E + \gamma ]}(H)$, with the error  controlled by
$ \norm{{\pa{H -E}\psi}}_{2}\le  \gamma \norm{\psi}_{2}$. (See  the derivation of \cite[Theorem~A.6]{GKloc} from  \cite[Theorem~A.1]{GKloc}.)
The notion  of   ``dominant  boxes",  introduced by Rojas-Molina and Veseli\' c  \cite[Subsection~5.2]{RV} (see also \cite[Appendix~A]{R}),  plays an important role in the derivation of Theorem~\ref{thmUCPSP} from Theorem~\ref{thmucp}.

  Using  Theorem~\ref{thmUCPSP}, we obtain (Theorems~\ref{thmWegfree} and \ref{thmWegfull})   optimal Wegner estimates  (i.e., with the correct dependence on the volume and interval length)  at all energies for  a class of non-ergodic  random  Schr\" odinger operators with alloy-type random potentials (called crooked Anderson Hamiltonians in Definition~\ref {defAndHam}).  As a consequence, we get optimal Wegner estimates  for  Delone-Anderson models at all energies (Remark~\ref{remDA}).
 We also prove (Theorem~\ref{thmWegbottom}) optimal Wegner estimates at the bottom of the spectrum for crooked Anderson Hamiltonians that have the expected dependence on the disorder  (in particular, the Wegner estimate improves as the disorder increases), a new result even for the usual (ergodic) Anderson Hamiltonian.   Using Theorem~\ref{thmWegbottom}, we prove localization at high disorder for Anderson Hamiltonians in a fixed interval at the bottom of the spectrum (Theorem~\ref{thmloc});  such a result  was previously known only with a covering condition  \cite[Theorem~3.1]{GKfinvol}.

We  use two norms on $\R^d$: \beq
\abs{x}=\abs{x}_2:=\pa{\sum_{j=1}^{d} \abs{x_{j}}^{2}}^{\frac 12 } \qtx{and} \abs{x}_\infty:=  \max_{j=1,2,\ldots,d} \abs{x_{j}}, 
\eeq
where $x=\pa{x_1,x_2,\ldots, x_d} \in \R^d$.
Distances between sets in $\R^{d}$ will be measured with respect to norm $\abs{x}$. The ball centered at $x\in \R^d$ with radius $\delta>0$ is given by
\beq
B(x,\delta):= \set{y \in \R^{d}; \, \abs{y-x}<\delta}. 
\eeq  
The  box (or cube) centered at $x\in \R^d$ with side of length  $L  $ is 
\beq
\La_L(x)= x + ]-\tfrac L 2, \tfrac L 2[^d=  \set{y \in \R^{d}; \, \abs{y-x}_\infty < \tfrac L 2};
\eeq
 we set
\beq
\widehat{\La}_L(x) = \La_L(x)\cap \Z^d.
\eeq
 Given subsets $A$ and $B$ of $\R^{d}$, and  a  function $\vphi$ on  the set $B$,  we set $\vphi_{A}:=\vphi \Chi_{A\cap B}$.  In particular, given $x\in   \R^{d}$ and $\delta >0$ we write $\vphi_{x,\delta}: =\vphi_{B(x, \delta )}$.   We let $\N_{\mathrm{odd}}$  denote the set of odd natural numbers.
If $K$ is an operator on a Hilbert space, $\cD(K)$ will denote its domain.   By a  constant  we will always mean a finite constant.  We will use  $C_{a,b, \ldots}$, $C^{\pr}_{a,b, \ldots}$,  $C(a,b, \ldots)$, etc., to  denote a constant depending only on the parameters
$a,b, \ldots$.

\begin{theorem}\label{thmUCPSP} Let  $H =   -\Delta + V$ be a  Schr\"odinger operator on 
$\mathrm{L}^2(\mathbb{R}^d)$,   where $V$ is a bounded potential.  Fix  $\delta \in ]0,\frac 1 {2}]$, let  $\set{y_k}_{k \in \Z^d}$ be sites in $\R^d$ with  $B(y_k,\delta) \subset \La_1(k)$ for all $k\in \Z^d$, and set
\beq\label{defWthm}
W= \sum_{k \in \Z^d } \Chi_{B(y_k,\delta) }.
\eeq
Given $E_{0}>0$, set $K =K(V,E_0)= 2\norm{V}_{\infty}+ E_{0} $.  Consider a box $\La=\La_L(x_0)$, where $x_0 \in \Z^d$ and $L\in \N_{\mathrm{odd}}$, $L \ge 72\sqrt{d}$.  There exists a constant $M_d>0$, such that, defining $\gamma =\gamma(d,
K,\delta) >0$ by
\beq\label{defgamma}
\gamma^2=\tfrac 1 2   \delta ^{M_d \pa{1 + K^{\frac 2 3}}} ,
\eeq
 then for any closed  interval $I \subset ]-\infty, E_{0}]$ with $\abs{I} \le 2\gamma $  we have
\beq \label{chivchi}
\Chi_{I}(H_\La) W \Chi_{I}(H_\La) \ge  \gamma^2  \Chi_{I}(H_\La).
\eeq
\end{theorem}

Theorem~\ref{thmUCPSP} is proved in Section~\ref{secUCPSP}. It  is derived from  the quantitative unique continuation principle  given in \cite[Theorem~3.2]{BKl}
using the  ``dominant  boxes"  introduced by Rojas-Molina and Veseli\' c  \cite[Subsection~5.2]{RV}, \cite[Appendix~A]{R}.

Combes, Hislop and Klopp  used  the UCPSP  to prove Wegner estimates for Anderson Hamiltonians, random  Schr\" odinger operators on $\L^2(\R^d)$ with  $q\Z^d$-periodic background potential ($q\in\N$) and  alloy-type random potentials located in the lattice $\Z^d$;  the estimate  \eq{UCPSP}  replaces the covering condition  required  by Combes and Hislop \cite{CH}. They obtained   optimal Wegner estimates    at all energies for these ergodic random  Schr\" odinger operators  \cite[Theorem~1.3]{CHK2}.

Rojas-Molina and Veseli\' c used \eq{RMV} to prove  Wegner estimates at all energies, optimal up to an additional factor of $ \abs{\log \abs{I}}^d$ ( $\abs{I}$ denotes the length of the interval $I$), for  a class of non-ergodic  random  Schr\" odinger operators on $\L^2(\R^d)$ with alloy-type random potentials, including Delone-Anderson models \cite[Theorem~4.4]{RV}. They also proved  optimal Wegner estimates at the bottom of the spectrum \cite[Theorem~4.11]{RV}

 These non-ergodic  random  Schr\" odinger operators are  `crooked' versions of the usual (ergodic) Anderson Hamiltonian. Theorem~\ref{thmUCPSP} leads to optimal Wegner estimates  at all energies for crooked  Anderson Hamiltonians. (In particular, we obtain optimal Wegner estimates  for  Delone-Anderson models at all energies; see Remark~\ref{remDA}.)

\begin{definition} \label{defAndHam} A \emph{crooked  Anderson Hamiltonian} is a  random Schr\"odinger  operator on 
$\mathrm{L}^2(\mathbb{R}^d)$ of the form
\beq\label{AndH}
H_{\bom}: =  H_0+
V_{\bom} ,
\eeq
where: 

\begin{enumerate}
\item  $H_0 = -\Delta + V\up{0}$,  where the  the background potential $ V\up{0}$  is bounded  and  $\inf \sigma(H_0)= 0 $.
\item $V_{\bom}$ is a \emph{crooked alloy-type random potential}:
\beq
V_{\bom} (x):= 
\sum_{j \in \Z^d} \omega_j   u_j(x),  \qtx{with} u_j(x)= v_j(x-y_j) ,\label{AndV}
\eeq
where, for some   $\delta_- \in ]0,\frac 12]$ and $ u_{-}, \delta_{+}, M\in ]0,\infty[$:
\begin{enumerate}
\item $\set{y_j}_{j \in \Z^d}$ are sites in $\R^d$ with  $B(y_j,\delta_-) \subset \La_1(j)$ for all $j\in \Z^d$;

\item the single site potentials $\set{v_j}_{j \in \Z^d} $  are measurable  functions on $\R^d$ with
 \begin{equation} \label{u}
 u_{-}\Chi_{B(0,\delta_-)}\le v_j \le \Chi_{\Lambda_{\delta_{+}}(0)}\qtx{for all} j \in \Z^d; \end{equation}
\item 

$\bom=\{ \omega_j\}_{j\in
\Z^d}$ is a family of independent 
 random
variables  whose  probability 
distributions $\{ \mu_j\}_{j\in
\Z^d}$ are non\--dege\-nerate with 
\beq \label{mu}
 \supp \mu_j \subset   [0,M ]  \qtx{for all} j \in \Z^d.
\eeq

\end{enumerate}
\end{enumerate}
\end{definition}

If  the background potential $V\up{0}$ is $q\Z^d$-periodic  with  $q\in\N$, and   $y_j=j$ and $v_j=v_0$ for all $j\in \Z^d$, then $H_{\bom}$ is the usual (ergodic) Anderson Hamiltonian.

Given a  crooked  Anderson Hamiltonian $H_{\bom} $,  we will use the following notation, definitions, and   observations:
\begin{itemize}
\item   We let $ V\up{0}_\infty:=\norm{ V\up{0}}_\infty $,  and set
\beq
U(x):= {\sum}_{j \in \Z^d}  u_j(x), \qtx{so}  U_\infty:= \norm{U}_\infty \le \pa{2+  \delta_+}^d .  \label{U+}
\eeq
\item We have 
 \beq
\norm{V_{\bom}}_\infty \le M U_\infty, \qtx{and hence} \norm{V\up{0} + V_{\bom}}_\infty \le  V\up{0}_\infty+ M U_\infty .
\eeq

\item  We set
\beq\label{defWU}
W := \sum_{j\in \Z^d } \Chi_{B(y_j,\delta_-) }  =\Chi_{\cup_{j\in \Z^d }B(y_j,\delta_-)},
\eeq
and note that
\beq\label{Wprop}
  0\le W  \le u_-^{-1} U,   \quad W^2=W, \qtx{and}  \norm{W}_\infty = 1.
\eeq

\item  We will consider  only boxes  $\La=\La_L(x_0)$, where $x_0 \in \Z^d$ and $L\in \N_{\mathrm{odd}}$.  For such a box $\La$  we define
finite volume crooked Anderson Hamiltonians, with either Dirichlet or periodic boundary condition, by
\beq
H_{\bom,\La}= H_{0,\La} +  V_{\bom}\up{\La} \qtx{on} \L^2(\La),
\eeq
where $H_{0,\La}$ is the restriction of $H_0$ to $\La$ with the specified boundary condition, and
\beq
V_{\bom}\up{\La}(x):=  \sum_{j \in \widehat{\La}} \omega_j   u_j(x) \qtx{for} x \in \R^d.
\eeq
We  also set 
\begin{align}
U\up{\La}(x)&:=  \sum_{j \in \widehat{\La}}  u_j(x) \le U(x) ,  \\
W\up{\La}(x)& :=  \sum_{j \in \widehat{\La}} \Chi_{B(y_j,\delta_-)} (x))\le u_-^{-1} U\up{\La}(x) , \label{Wprop22}
\end{align} 
and note that  $W\up{\La}(x)=W(x)$ for $x \in \La$.

\item   We write   $P_{\bom,\La}(B):=\Chi_{B}(H_{\bom,\La})$ for a Borel set $B \subset \R$.

\item   Given a box $\La$, we set $S_\La(t) := \max_{j \in \widehat\La}    S_{\mu_j}(t) $ for $t\ge 0$, where $S_{\mu}(t):=  \sup_{a \in \R} \mu([a,a+t]) $ denotes  the concentration function of the probability measure $\mu$.  We also set  $S(t) := \sup_{j \in \Z^d}    S_{\mu_j}(t) $ for $t\ge 0$.

\end{itemize}

\begin{remark}  We defined a \emph{normalized} crooked  Anderson Hamiltonian.  Requiring  $\inf \sigma(H_0)= 0 $ is just a convenience.  It suffices to have   $ v_j \le u_+\qtx{for all} j \in \Z^d$ for some $u_+ \in ]0,\infty[$  in \eq{u} (we took $u_+=1$), and  we need only 
$\supp \mu_j \subset   [M_-,M_+ ]  \qtx{for all} j \in \Z^d$ with $M_\pm \in \R$ in \eq{mu}.  Since an unrenormalized crooked  Anderson Hamiltonian is always equal to a renormalized crooked  Anderson Hamiltonian plus a constant (see the argument  in \cite[Subsection~2.1]{GKloc}), there is no loss of generality in taking $H_{\bom}$ as in  Definition~\ref{defAndHam}.
\end{remark}

Let $H_{\bom}$  be a crooked Anderson Hamiltonian $H_{\bom}$.  Using the UCPSP of Theorem~\ref{thmUCPSP} with $H=H_{0}$ and $W$ as in \eq{defWU},  we can simply  follow the proof in   \cite{CHK2} obtaining the following  extension of their results for crooked  Anderson Hamiltonians.  

 \begin{theorem}\label{thmWegfree} Let $H_{\bom}$ be a crooked Anderson Hamiltonian.
Given $E_{0}>0$, set $K_0=E_{0} + 2V\up{0}_\infty $, and define $\gamma_0 =\gamma_0(d,
K_0,\delta_-) >0$ by
\beq\label{defgamma0}
\gamma_0^2=\tfrac 1 2   \delta_- ^{M_d \pa{1 + K_0^{\frac 2 3}}},
\eeq
where $M_d>0$ is the constant of Theorem~\ref{thmUCPSP}.
 Then for any closed  interval $I \subset ]-\infty,E_{0}]$ with $\abs{I} \le 2\gamma_0 $ and  any box $\La=\La_L(x_0)$, where $x_0 \in \Z^d$ and $L\in \N_{\mathrm{odd}}$, $L \ge 72\sqrt{d}+ \delta_+$,   we have
\beq\label{opWegner0}
\E\set{ \tr P_{\bom,\La }(I) } \le  C_{d, \delta_\pm,u_-, V\up{0}_\infty,E_0}
  \pa{1+M^{2^{2 + \frac {\log d}{\log 2}}}}S_\La(\abs{I})\abs{\La}.
\eeq
\end{theorem}

We may also use Theorem~\ref{thmUCPSP} with $H=H_0 +V\up{\La}_{\bom}$ and $W$ as in \eq{defWU},  obtaining the  UCPSP \eq{chivchi} with a constant $\gamma$ independent of $\bom$.
  In Lemma~\ref{lemWeg} we show how this implies a Wegner estimate.   Combining Theorem~\ref{thmUCPSP} and Lemma~\ref{lemWeg}  yields the following optimal Wegner estimate.

 \begin{theorem}\label{thmWegfull}  Let $H_{\bom}$ be a crooked Anderson Hamiltonian.
Given $E_{0}>0$, set $K=E_{0} + 2\pa{V\up{0}_\infty+MU_\infty}$, and define $\gamma =\gamma(d,
K,\delta_-) >0$ by
\beq\label{defgammaW}
\gamma^2=\tfrac 1 2   \delta_- ^{M_d \pa{1 + K^{\frac 2 3}}},
\eeq
where $M_d>0$ is the constant of Theorem~\ref{thmUCPSP}.
 Then for any closed  interval $I \subset ]-\infty,E_{0}]$ with $\abs{I} \le 2\gamma $ and  any box $\La=\La_L(x_0)$, where $x_0 \in \Z^d$ and $L\in \N_{\mathrm{odd}}$, $L \ge 72\sqrt{d}+ \delta_+$,   we have
\beq\label{opWegner}
\E\set{ \tr P_{\bom,\La }(I) } \le  C_{d, \delta_+, V\up{0}_\infty}
   \pa{u_-^{-2}\gamma^{-4} (1 +E_0)}^{2^{1 + \frac {\log d}{\log 2}} }S_\La(\abs{I})\abs{\La}.
\eeq
\end{theorem}

Theorems~\ref{thmWegfree} and \ref{thmWegfull} are proved in Section~\ref{secWeg}.
They both give optimal Wegner estimates valid at all energies,  but the constants in \eq{opWegner0} and \eq{opWegner} differ on their dependence on the relevant parameters.

\begin{remark}[The  Delone-Anderson model] \label{remDA}  Theorems~\ref{thmWegfree} and \ref{thmWegfull}  can be applied to the Delone-Anderson model, improving the Wegner estimate of \cite[Theorem~4.4]{RV}.  The Delone-Anderson Hamiltonian is defined almost exactly as in Definition~\ref{defAndHam}, the  difference being that the crooked alloy-type random potential of \eq{AndV} is replaced by the Delone-Anderson random potential
\beq\label{DelAndV}
V_{\bom} (x):= 
\sum_{j \in \Dd} \omega_j   u_j(x),  \qtx{with} u_j (x)= v_j(x-j), 
\eeq
where:
\begin{enumerate}
\item $\Dd\subset \Z^d$ is a Delone set, i.e., there exist scales  $0< K_1 < K_2$ such that $\#\pa{ \Dd \cap \La_{K_1}(x) }\le 1$ and $\# \pa{\Dd \cap \La_{K_2}(x)} \ge 1$ for all $x\in \R^d$, where $\#A$ denotes the cardinality of the set $A$;

\item $\bom=\set{\omega_j}_{j\in \Dd}$ and $\set{v_j}_{j\in \Dd}$ are as in 
Definition~\ref{defAndHam} with $\Dd$
 substituted for $\Z^d$.
\end{enumerate}
 We set
 $R=2 \min \set{r \in \N; \; r \ge \frac {K_2}2  +\delta_-}$, and   fix  
 $\, y_k \in \Dd \cap \La_{K_2}(k)$ for each $k\in R Z^d$; note that  $B(y_k,\delta_-)\subset \La_R (k)$.  We set
 $\Dd_1=\set{y_k}_{k\in R Z^d}$ and $\Dd_2=\Dd\setminus \Dd_1$, and   decompose the Delone-Anderson random potential  similarly to \cite[Eq.~(21)]{RV}:
 \begin{gather}
 V_{\bom} (x)= V_{\bom\up{1}} (x) + V_{\bom\up{2}} (x),\\ \notag
 \text{where} \quad  \bom\up{i}=\set{\omega_j}_{j\in \Dd_i}\qtx{and} V_{\bom\up{i}} (x):= 
\sum_{j \in \Dd_i} \omega_j   u_j(x) \qtx{for} i=1,2.
 \end{gather}
Note that $V_{\bom\up{2}} \ge 0$, and, 
 since 
 $\Dd$ is a Delone set, there exists a constant $V\up{2}_\infty$ such that
 $\norm{V_{\bom\up{2}}}_\infty\le V\up{2}_\infty$ for $\P$-a.e.\ $\bom\up{2}$.  We set  
 \beq
 H_{\bom\up{1}}\up{\bom\up{2}}: =-\Delta + V\up{0,\bom\up{2}} +V_{\bom\up{1}}, \qtx{where}  V\up{0,\bom\up{2}}= V\up{0}+ V_{\bom\up{2}},
 \eeq
and note that 
\beq\label{V2unif}
\norm{V\up{0,\bom\up{2}} }_\infty\le V_\infty\up{0} +  V\up{2}_\infty \qtx{for} \P\text{-a.e.} \ \bom\up{2}.
\eeq 
 If we had  $R=1$,  $H_{\bom\up{1}}\up{\bom\up{2}}$
 would be a crooked Anderson Hamiltonian  with background potential $V\up{0,\bom\up{2}}$
and alloy-type potential $V\up{1}_{\bom\up{1}}$,   but 
would not be not normalized as in Definition~\ref{defAndHam} since we we only have
$\inf \sigma \pa{-\Delta+ V\up{0,\bom\up{2}} }\ge 0$.  But Theorems~\ref{thmWegfree} and \ref{thmWegfull}  hold as stated  with the same constants if we only required $\inf \sigma(H_0)\ge 0 $ in Definition~\ref{defAndHam}.  Moreover, Theorems~\ref{thmUCPSP}, \ref{thmWegfree} and \ref{thmWegfull} are valid  with boxes of side $R$ instead of boxes of side $1$, except that all the constants would depend on $R$.  We can thus apply Theorems~\ref{thmWegfree} and \ref{thmWegfull}, averaging only  with respect to $\bom\up{1}$,  to obtain Wegner estimates for $H_{\bom\up{1}}\up{\bom\up{2}}$ with $S_\La(t) := \max_{j \in \Dd_1 \cap\La}    S_{\mu_j}(t) $,  with   constants  independent of $\bom\up{2}$  for $\P$-a.e.\ $\bom\up{2}$
in view of \eq{V2unif}. We thus conclude that the Wegner estimates of  Theorems~\ref{thmWegfree} and \ref{thmWegfull}
are valid for the Delone-Anderson model, with $V_\infty\up{0} +  V\up{2}_\infty$ substituted for  $V_\infty\up{0}$ and the constants  also depending on  the scale $R$.
\end{remark}

The constants in the Wegner estimates  \eq{opWegner0} and \eq{opWegner} grow fast with the disorder.  To see that, consider $H_{\bom,\lambda}= H_0 + \lambda V_{\bom} $, where $H_0$ and $ V_{\bom} $ are as in Definition~\ref{defAndHam} and $\lambda >0$ is the disorder parameter. $H_{\bom,\lambda}$ can be rewritten as a crooked Anderson Hamiltonian $H\up{\lambda}_{\bom}= H_0 + V_{\bom} $ in the form of Definition~\ref{defAndHam} by replacing the probability distributions $\set{\mu_j}_{j \in \Z^d}$ by the probability distributions $\set{\mu\up{\lambda}_j}_{j \in \Z^d}$, where
$\mu\up{\lambda}_j$ is the probability distribution of the random variable $\lambda \omega_j$, that is, 
\beq\label{defmul}
\mu\up{\lambda}_j(B) = \mu_j(\lambda^{-1}B) \qtx{for all Borels sets} B\subset \R.
\eeq
We clearly have $S_{\mu\up{\lambda}_j}(t) =S_{\mu_j}(\frac t{\lambda})$, and 
it follows from \eq{mu} that
\beq \label{mulambda}
 \supp\mu\up{\lambda}_j \subset   [0,M_\lambda ],  \qtx{where }  M_\lambda= \lambda M.
 \eeq
 Applying the  Wegner estimates 
\eq{opWegner0} and \eq{opWegner} to $H_{\bom,\lambda}$  we get
(we omit the dependence on the constants from  Definition~\ref{defAndHam})
\begin{align}
\label{opWegner011}
\E\set{ \tr P_{\bom,\lambda,\La }(I) } &\le  C_{E_0}
 \pa{1+ \lambda^{2^{2 + \frac {\log d}{\log 2}}}}S_\La(\lambda^{-1}\abs{I})\abs{\La} \;\;  &\text{from \eq{opWegner0}},\\ \label{opWegner11}
 \E\set{ \tr P_{\bom,\lambda,\La }(I) } &\le  C_{E_0} \e^{c_{E_0}\pa{1+\lambda^{\frac 2 3}}}
  S_\La(\lambda^{-1}\abs{I})\abs{\La}   & \text{from \eq{opWegner}}.
\end{align}
The constants in these  Wegner estimates grow as the disorder increases.

The Wegner estimate  \eq{opWegner011} is what one gets for the usual Anderson Hamiltonian from  \cite{CHK2} without further assumptions.
But if the crooked  Anderson Hamiltonian satisfies the covering condition
$U\up{\La}\ge \alpha \Chi_{\La}$ for some  $\alpha>0$,
the UCPSP \eq{UCPSP} holds trivially on $\L^2(\La)$ for all intervals $I$ with $H=H_{0,\La}$ or $H=H_{\bom,\La}$,  $W=U\up{\La}$, and $\kappa=\alpha$, so, either  proceeding as in \cite{CH} if we use \eq{UCPSP} with $H=H_0$, or using Lemma~\ref{lemWeg} if we take $H=H_\bom$ in \eq{UCPSP},  we get an  optimal Wegner estimates of the form
\beq\label{opWegnercovcond}
\E\set{ \tr P_{\bom,\La }(I) } \le  C_{d, \delta_+,\alpha, V\up{0}_\infty,E_0}
  S_\La(\abs{I})\abs{\La}.
\eeq
Note that the constant does not depend on $M$, so introducing the disorder parameter $\lambda$ we get
\beq\label{opWegnercovcondlamb}
 \E\set{ \tr P_{\bom,\lambda,\La }(I) }  \le  C_{d, \delta_+,\alpha, V\up{0}_\infty,E_0}
  S_\La(\lambda^{-1}\abs{I})\abs{\La}.
\eeq
In other words, the constant in the  Wegner estimate improves as the disorder increases.

Up to now  an estimate like \eq{opWegnercovcondlamb}   had not been proven for Anderson Hamiltonians without  the covering condition.  While we are not able to  prove this estimate at all energies without the covering condition, we can prove them at the bottom of the the spectrum, a new result even for the usual (ergodic) Anderson Hamiltonian.

We write 
$H_\La\up{D}$ to denote the  restriction of  a Schr\"odinger operator $H$   to    the  box $\La$ with Dirichlet  boundary condition, and set $P\up{D}_{\La}(B):=\Chi_{B}(H\up{D}_{\La})$.  We recall that Dirichlet boundary condition
implies  $\inf \sigma(H_{\La}\up{D}) \ge \inf \sigma(H)$.

Given  a crooked Anderson Hamiltonian $H_{\bom}$, we define finite volume operators 
$H\up{D}_{\bom,\La } = H_{0,\La}\up{D} + V_{\bom}\up{\La}$, and let $P\up{D}_{\bom,\La}(B):=\Chi_{B}(H\up{D}_{\bom,\La } )$.
We set $H(t)= H_0 + t u_-W$ for $t\ge 0$, and note
\beq\label{Dirichletbound}
0\le E(t): = \inf \sigma(H(t)) \le E\up{D}_\La(t):= \inf \sigma(H_{\La}\up{D}t)).
\eeq
  By our normalization  $E(0)=0$, and it follows from the min-max principle  that  $0\le E(t_2)- E(t_1)\le(t_2- t_1)u_-$ for $0\le t_1 \le t_2$.
We may thus  define  
\beq
E(\infty):=\lim_{t\to \infty} E(t)=\sup
_{t\ge 0} E(t)  \in [0,\infty].
\eeq
 If  $W=I$ we have  $E(\infty)=\infty$.  But if not, that is, if   $\Upsilon =\R^d \setminus  \overline{\cup_{j\in \Z^d }B(y_j,\delta_-)}\not= \emptyset $, letting  $H\up{D}_{0,\Upsilon}$ denote the restriction of $H_0$ to  $\Upsilon$ with Dirichlet boundary condition, we get
\beq
E(t) \le E(\Upsilon):= \inf \sigma(H\up{D}_{0,\Upsilon})< \infty \;\;\text{for}\;\;  t \ge 0 \quad \Longrightarrow
\quad E(\infty) \le E(\Upsilon)<\infty.
\eeq

More importantly, Rojas-Molina and Veseli\' c proved that $E(\infty)>0$ \cite[Theorem~4.9]{RV}, \cite[Theorem~A.3.1]{R}.
By a similar argument, we establish strictly positive lower bounds for  $E(t)$ and $E(\infty)$ in Lemma~\ref{lemEinfty}.

 \begin{theorem}\label{thmWegbottom} Let $H_{\bom}$ be a crooked Anderson Hamiltonian.    Then $E(\infty) >0$. Let $E_{1}\in ]0, E(\infty)[$, so  we have 
 \beq\label{kappabottom}
\kappa = \kappa(H_0,u_-W,E_1)= \sup_{s>0; \; E(s) >E_1} \frac {E(s) -E_{1}}s   >0,
\eeq
and consider a box   $\La=\La_L(x_0)$ with $x_0 \in \Z^d$ and $L\in \N_{\mathrm{odd}}$, $L \ge 2 + \delta_+$.   Then
 \beq \label{PVPb}
P\up{D}_{\bom,\La }( ]-\infty,E_{1}])U\up{\La}P\up{D}_{\bom,\La }( ]-\infty,E_{1}]) \ge \kappa P\up{D}_{\bom,\La }( ]-\infty,E_{1}]),
\eeq
and for any closed interval $I \subset ]-\infty,E_{1}]$ we have
\begin{align}\label{preWegner}
\E\set{ \tr P\up{D}_{\bom,\La }(I) }  \le  C_{d, \delta_+, V\up{0}_\infty}
   \pa{\kappa^{-2} (1 +E_1)}^{2^{1 + \frac {\log d}{\log 2}} }S_\La(\abs{I})\abs{\La}.
    \end{align}
    In particular, for all disorder $\lambda >0$ we have
  \begin{align}\label{preWegnerdis}
\E\set{ \tr P\up{D}_{\bom,\lambda,\La }(I) }  \le  C_{d, \delta_+, V\up{0}_\infty}
   \pa{\kappa^{-2} (1 +E_1)}^{2^{1 + \frac {\log d}{\log 2}} }S_\La(\lambda^{-1}\abs{I})\abs{\La}.
    \end{align}  
   for any closed interval $I \subset  ]-\infty,E_{1}]$  
\end{theorem}

Theorem~\ref{thmWegbottom} is proven in Section~\ref{secbottom}.  We use Lemma~\ref{lemBLS},  a slight extension of an abstract UCPSP due to Boutet de Monvel, Lenz, and Stollmann  \cite [Theorem~1.1]{BLS},   to prove \eq{PVPb}. The estimate \eq{preWegner} then follows from Lemma~\ref{lemWeg}.   Since $\kappa$ in \eq{kappabottom} does not depend on $M$,  Lemma~\ref{lemWeg} gives a constant in the Wegner estimate \eq{preWegner}  independent of $M$, so \eq{preWegnerdis} follows.

  Theorem~\ref{thmWegbottom} is the missing link for proving localization at high disorder for Anderson Hamiltonians in a fixed interval at the bottom of the spectrum.  This was previously known only with a covering condition $U\up{\La}\ge \alpha \Chi_{\La} $, where $\alpha >0$ \cite[Theorem~3.1]{GKfinvol}.   
  
  We state the theorem in the generality  of crooked Anderson Hamiltonians.   (The bootstrap  multiscale analysis can be adapted for crooked  Anderson Hamiltonians \cite{R1,R}.)  By complete localization on an interval $I$ we mean that for all $E\in I$ there exists $\delta(E)>0$ such that  we can perform the bootstrap multiscale analysis on the interval $(E-\delta(E),E+\delta(E))$, obtaining Anderson and dynamical localization; see \cite{GKboot,GKfinvol,GKsudec}.

\begin{theorem}\label{thmloc} Let $H_{\bom,\lambda}$ be a crooked Anderson Hamiltonian with disorder $\lambda>0$, and suppose  the single-site probability distributions $\{ \mu_j\}_{j\in
\Z^d}$ satisfy  $S(t) := \sup_{j \in \Z^d}    S_{\mu_j}(t) \le C t^\theta$ for all  $t\ge 0$, where $\theta \in ]0,1]$ and $C$ is a constant.   Given  $E_{1}\in ]0, E(\infty)[$,  there exists $\lambda(E_1) <\infty$ (depending also on $d,V\up{0}_\infty, u_-,\delta_\pm, U,\theta,C$), such that $H_{\bom,\lambda}$ exhibits complete localization on the interval $[0, E_1[$ for all $\lambda \ge \lambda(E_1)$.
\end{theorem}

Theorem~\ref{thmloc}  is proven in Section~\ref{secbottom}.

\section{Unique continuation principle for spectral projections}\label{secUCPSP}

In this section we prove  Theorem~\ref{thmUCPSP}. We start by recalling the quantitative unique continuation principle as given in \cite[Theorem~3.2]{BKl}.

\begin{theorem}\label{thmucp} Let  ${\Omega}$ be an  open subset  of $\R^d$ and consider    a real measurable function $V$ on ${\Omega}$ with $\norm{V}_{\infty} \le K <\infty$.   Let 
$\psi \in\mathrm{H}^2({\Omega})$ be real valued and  let  ${\zeta} \in \L^2({\Omega})$  be defined by
\beq \label{eq}
-\Delta {\psi} +V{\psi}={\zeta}  \qtx{a.e.\  on} \Omega.
\eeq
 Let   ${\Theta} \subset {\Omega}$  be a bounded measurable set where $\norm{\psi_{\Theta}}_2 >0$. 
Set
\beq \label{defRx0}
{Q}(x,\Theta):= \sup_{y \in \Theta } \abs{y - x} \qtx{for} x \in {\Omega}.
\eeq
Consider $x_0 \in {\Omega}\setminus \overline{\Theta}$ such that
\beq
  \label{xR}
{Q}={Q}(x_0,\Theta)\ge  1 \qtx{and} B(x_0, 6{Q}+ 2)\subset {\Omega}.
\eeq
Then, given
\beq \label{delta}
0<  \delta \le \min\set{   \dist \pa{x_0, {\Theta}},\tfrac 1 {2}},
\eeq
we have
\begin{align} \label{UCPbound}
 \pa{\frac \delta{Q}}^{m_d \pa{1 + K^{\frac 2 3}}\pa{ {Q}^{\frac 43}  +  \log \frac{\norm{ {\psi}_{{\Omega}}}_{2}} {\norm{ {\psi}_{{\Theta}}}_2}}}\norm{ {\psi}_{{\Theta}}}^2_2  \le   \norm{ {\psi}_{x_0,\delta}}^2_2 + \delta^2 \norm{{\zeta_{{\Omega}}}}_2^2,
\end{align}
where $m_d>0$ is a constant depending only on $d$.
\end{theorem}

 Note the condition $\delta \le \frac 1 {2}$ in \eq{delta} instead of $\delta \le \frac 1 {24}$ as in  \cite[Eq.~(3.2)]{BKl}.  All that is needed  in \eq{delta} is an upper bound $\delta\le \delta_0$;  the constant $m_d$ in  \eq{UCPbound}  then depending on $\delta_0$. 

Note that  for  $\psi \in \L^2(\Lambda)$ we have  $\psi=\psi_\Lambda$ in our notation, and hence $ \norm{\psi}_2= \norm{\psi_\Lambda}_2$.

\begin{theorem} \label{lemUCPeig} Let  $H =   -\Delta + V$ be a  Schr\"odinger operator on 
$\mathrm{L}^2(\mathbb{R}^d)$, where V is a bounded potential with $\norm{V}_{\infty}\le K$.  Fix  $\delta \in ]0,\frac 1 {2}]$, let  $\set{y_k}_{k \in \Z^d}$ be sites in $\R^d$ with  $B(y_k,\delta) \subset \La_1(k)$ for all $k\in \Z^d$. 
 Consider a box $\La=\La_L(x_0)$, where $x_0 \in \Z^d$ and $L\in \N_{\mathrm{odd}}$, $L \ge 72\sqrt{d}$.   Then for all  real-valued $\psi \in \cD(\Delta_\La) $    
we have
\begin{align} \label{UCPdelta}
 \delta^{M_d \pa{1 + K^{\frac 2 3}}} \norm{\psi_\Lambda}_2^2   \le  
\sum_{k \in \widehat{\La}}\norm{ {\psi}_{y_{k},\delta}}^2_2 +  \delta^2 \norm{\pa{(-\Delta + V)\psi}_\Lambda}_2^2,
\end{align}
where $M_d>0$ is a constant depending only on $d$.
\end{theorem}

\begin{proof}Without loss of generality we take $x_0=0$, so $\La=\La_L(0)$ with  $L\in \N_{\mathrm{odd}}$, $L \ge 72\sqrt{d}$.
As in \cite[Proof of Corollary~A.2]{GKloc}, we extend $V$ and  functions $\vphi \in \mathrm{L}^2(\Lambda)$ to $\R^d$ as follows. \smallskip

\noindent{\textbf{Dirichlet boundary condition:}}
Given $\vphi \in \mathrm{L}^2(\Lambda)$, we extend it to a function $\widetilde{\vphi}\in \mathrm{L}^2_{\mathrm{loc}}(\R^d)$ by setting $\widetilde{\vphi}=\vphi$ on $\Lambda$ and $\widetilde{\vphi}=0 $ on $\partial \Lambda$, and requiring 
\beq \label{widetildephi}
\widetilde{\vphi}(x)=- \widetilde{\vphi}(x + (L  -2 \widehat{x_j}) \e_j)\qtx{for all} x \in \R^d \qtx{and} j \in \set{1,2 \dots,d},
\eeq
where $\set{\e_j}_{j =1,2\ldots,d}$ is the canonical orthonormal basis in $\R^d$, and for each  $t\in  \R$ we define $\hat{t}\in  ] -\frac L 2, \frac L 2]$ by $t =
kL + \hat{t}$ with $k \in \Z$.  We also  extend the potential $V$ to a potential $\widehat{V}$ on  $\R^d$ by  by setting $\widehat{V}=V$ on $\Lambda$ and $V=0$ on $\partial \Lambda$, and requiring that 
for all $x \in \R^d$ and $j \in \set{1,2 \dots,d}$ we have
\beq
\widehat{V}(x)=\widehat{V}(x + (L  -2 \widehat{x_j}) \e_j).
\eeq
Note that  $\|\widehat{V}\|_\infty =\norm{V}_\infty \le K$.  Moreover, 
$\psi \in \D(\Delta_\Lambda)$  implies   $\widetilde{\psi} \in \mathrm{H}^2_{\mathrm{loc}}(\R^d)$ and
\beq\label{eqwidehat}
\widetilde{(-\Delta +V) {\psi} }=   (-\Delta  + \widehat{V} ) \widetilde{\psi} .
\eeq

\noindent{\textbf{Periodic boundary condition:}}
We extend  $\vphi \in \mathrm{L}^2(\Lambda)$  and $V$ to  periodic functions $\widetilde{\vphi}$ and $\widehat{V}$ on $\R^d$ of period $L$;  note  $\|\widehat{V}\|_\infty =\norm{V}_\infty \le K$.  Moreover, 
$\psi \in \D(\Delta_\Lambda)$  implies   $\widetilde{\psi} \in \mathrm{H}^2_{\mathrm{loc}}(\R^d)$ and we have  \eq{eqwidehat}.
\smallskip

 We now  take  $Y \in \N_{\mathrm{odd}}$, $Y <\frac L 2$  (to be specified later), and note that since $L$ is odd, we have 
\beq\label{Ladecomp}
\overline{\La}=\bigcup_{k \in \widehat{\La}} \overline{\La_1(k)}.
\eeq
 It follows that  for  all $ \vphi \in \mathrm{L}^2(\Lambda)$ we have (see  \cite[Subsection~5.2]{RV})
\beq\label{sumY}
\sum_{k \in \widehat{\La}} \norm{\widetilde{\vphi}_{\La_Y(k)}}_2^2
\begin{cases} \le (2Y)^d \norm{\vphi_\La}_2^2 & \text{for Dirichlet boundary condition}\\
= Y ^d \norm{\vphi_\La}_2^2 & \text{for periodic boundary condition}
\end{cases}.
\eeq

We now fix   $\psi \in \cD(\Delta_\La) $.   Following Rojas-Molina and Veseli\' c, we  call a site $k \in \widehat{\La}$ \emph{dominating} (for $\psi$) if
\beq\label{dom}
\norm{\psi_{\La_1(k)}}_2^2 \ge\tfrac 1{2 (2Y)^{d}} \norm{\widetilde{\psi}_{\La_Y(k)}}_2^2.
\eeq
Letting  $\widehat{D}\subset \widehat{\La}$ denote the collection of dominating sites, Rojas-Molina and Veseli\' c   \cite[Subsection~5.2]{RV} observed that it follows from \eq{sumY},  \eq{dom}, and \eq{Ladecomp},  that
\beq\label{sumdom}
\sum_{k \in \widehat{D}} \norm{\psi_{\La_1(k)}}_2^2 \ge \tfrac 1 2 \norm{\psi_\La}_2^2.  
\eeq

We  define a map $J\colon  \widehat{D} \to \widehat{\La}$ by
\beq \label{defJ}
 J(k)=
\begin{cases}   k + 2 \e_1 \qtx{if}  k + 2 \e_1\in \widehat{\La}\\
k -  2 \e_1 \qtx{if}  k + 2 \e_1\notin \widehat{\La}
\end{cases}.
\eeq
Note that $J$ is well defined, 
\beq\label{Jinv}
\# J^{-1} (\set{j}) \le 2 \qtx{for all} j \in \widehat{\La},
\eeq
and  recalling \eq{defRx0},
\beq\label{Qest}
  {Q}(y_{J(k)},\La_1(k)) = \tfrac 1 2  \sqrt{24 + d} \le  \tfrac 5 2  \sqrt{d} \qtx{for all} k \in \widehat{D}.
\eeq

Choosing
\begin{align}\label{Y36}
Y= \min\set{n \in   \N_{\mathrm{odd}}; n >  2\pa{\pa{2 +\tfrac {\sqrt{d}} 2} + \pa{3 \sqrt{24 + d}  + 2}}}\le   40\sqrt{d},
\end{align}
we have  $Y < \frac L 2$ and 
\beq
B(y_{J(k)}, 6{Q}(y_{J(k)},\La_1(k))+2)\subset \La_Y(k)  \qtx{for all} k \in \widehat{D}.
\eeq
For each $k\in \widehat{D}$  we may thus   apply Theorem~\ref{thmucp}   with $\Omega= \La_Y(k) $ and $\Theta= \La_1(k) $, using \eq{Qest} and \eq{dom}, obtaining    
\begin{align} \label{UCPbound3}
 \delta^{m^{\prime}_d\pa{1 + K^{\frac 2 3}}}\norm{ {\psi}_{{\La_1(k)}}}^2_2  \le   \norm{ {\psi}_{y_{J(k)},\delta}}^2_2 + \delta^2 \norm{{\widetilde{\zeta}_{\La_Y(k) }}}_2^2,
\end{align}
where  $\zeta= (-\Delta +V)\psi$ and  $m^{\prime}_d>0$ is a constant depending only on $d$.  Summing over $k\in \widehat{D}$ and using \eq{sumdom}, \eq{Jinv},   \eq{sumY}, and \eq{Y36},  and we get
\begin{align}
\tfrac 1 2  \delta^{m^{\prime}_d\pa{1 + K^{\frac 2 3}}}\ \norm{\psi_\La}_2^2 &\le  
2\sum_{k \in \widehat{\La}}\norm{ {\psi}_{y_{k},\delta}}^2_2 +  (2Y)^d \delta^2 \norm{{{\zeta_\La}}}_2^2\\
& \le 2\sum_{k \in \widehat{\La}}\norm{ {\psi}_{y_{k},\delta}}^2_2 +  (80\sqrt{d})^d \delta^2 \norm{{{\zeta_\La}}}_2^2,  \notag\end{align}
so \eq{UCPdelta} follows.
\end{proof}

\begin{comment}  The final version of  \cite {RV}  uses a map similar to \eq{defJ}, see \cite[Subsection~5.3]{RV}.
\end{comment}

We are now ready to prove  Theorem~\ref{thmUCPSP}.

\begin{proof}[Proof of  Theorem~\ref{thmUCPSP}] Given $E_{0}>0$, set $K =K(V,E_0)= 2\norm{V}_{\infty}+ E_{0} $, and   let $\gamma$ be given by \eq{defgamma}, where $M_d>0$ is the constant in Theorem~\ref{lemUCPeig}.  Let $I \subset ]-\infty, E_{0}]$ be a closed interval with $\abs{I} \le 2\gamma $.  Since  $\sigma (H_\La)\subset [-\norm{V }_\infty,\infty[$ for any box $\La$,  without loss of  generality we assume $I=[E-\gamma ,E + \gamma ]$ with $E \in [ -\norm{V }_\infty,E_{0}]$, so  \beq
\norm{V -E}_\infty \le \norm{V }_\infty + \max \set{E_0, \norm{V }_\infty}\le K.
\eeq
Moreover, for any 
box $\La$  we have  
 \beq\label{IgammaI}
   \norm{{\pa{H_\La -E}\psi}}_{2}\le  \gamma \norm{\psi}_{2} \qtx{for all} \psi \in \Ran  {\Chi}_I(H_\La).
  \eeq

Let $\La$ be a box as in Theorem~\ref{lemUCPeig} and $\psi \in \Ran  {\Chi}_I(H_\La)$. If $\psi$  is real-valued,  it follows from Theorem~\ref{lemUCPeig}, \eq{defgamma}, and \eq{IgammaI} that
\begin{align} \label{UCPdelta2}
 2\gamma^2\norm{\psi}_2^2   \le  
\sum_{k \in \widehat{\La}}\norm{ {\psi}_{y_{k},\delta}}^2_2 +  \gamma^2 \norm{\psi}_2^2, 
\end{align}
yielding \beq\label{UCPdelta248}
\gamma^2\norm{\psi}_2^2   \le \sum_{k \in \widehat{\La}}\norm{ {\psi}_{y_{k},\delta}}^2_2= \norm{{W \psi}}_2^2,
\eeq
where the equality follows from  \eq{defWthm}.  For arbitrary $\psi \in \Ran  {\Chi}_I(H_\La)$, we write
$\psi = \Rea \psi + i\Ima \psi$, and note that $\Rea \psi, \Ima \psi  \in \Ran  {\Chi}_I(H_\La)$, $\norm{\psi}_2^2= \norm{\Rea \psi}_2^2 + \norm{\Ima\psi}_2^2$, and, since $W$ is real-valued, $\norm{W\psi}_2^2= \norm{W\Rea \psi}_2^2 + \norm{W\Ima\psi}_2^2$.  Recalling     $W^2=W$, we conclude that
\beq\label{UCPdelta2489}
\gamma^2 \scal{\psi,\psi}=\gamma^2\norm{\psi}_2^2   \le  \norm{{W \psi}}_2^2= \scal{\psi,W \psi}
\eeq
for all $\psi \in \Ran  {\Chi}_I(H_\La)$, proving \eq{chivchi}.  
\end{proof}

\section{Wegner estimates}\label{secWeg}

In this section we prove Theorems~\ref{thmWegfree} and \ref{thmWegfull}.  

Note that for a   crooked Anderson Hamiltonian $H_\bom$ and a box $\La$, we always  have
 \beq\label{finsp}
 \sigma(H_{0,\La})\subset [- \alpha ,\infty[ \qtx{and} \sigma(H_{\bom,\La})\subset [-\alpha ,\infty[,
 \eeq
where  $\alpha=0$ for Dirichlet boundary condition and $\alpha= V\up{0}_\infty$ for periodic boundary condition.

\begin{proof}[Proof of Theorem~\ref{thmWegfree}]

Let $H_\bom$ be a be a crooked Anderson Hamiltonian.
Given $E_{0}>0$, set $K_0=E_{0} + 2V\up{0}_\infty $, and define $\gamma_0$ by \eq{defgamma0}.   We apply Theorem~\ref{thmUCPSP} with $H=H_{0}$ and $W$ as in \eq{defWU}, concluding that for 
 any closed  interval $I \subset ]-\infty ,E_{0}]$ with $\abs{I} \le 2\gamma_0 $ and  any box $\La$ as in the hypotheses of the theorem, we have, using also \eq{Wprop},
 \beq \label{chivchiWegfree}
 \Chi_{I}(H_{0,\La})\le  \gamma_0^{-2} \Chi_{I}(H_{0,\La}) W\up{\La} \Chi_{I}(H_{0,\La}) \le u_-^{-1}  \gamma_0^{-2}\Chi_{I}(H_{0,\La}) U\up{\La}  \Chi_{I}(H_{0,\La}).
\eeq

In view of \eq{finsp}, it suffices to 
take $I\subset  [-\alpha ,E_0]$.  We can now  follow the proof in   \cite{CHK2}, using \eq{chivchiWegfree} instead of \cite[Theorem~2.1]{CHK2}, and keeping careful track of the dependence of the constants on the relevant parameters,  obtaining  \eq{opWegner0}.
\end{proof}

We now turn to the proof of Theorem~\ref{thmWegfull}.  We start by showing that, given
 a crooked  Anderson Hamiltonian $H_{\bom}$,  the  UCPSP \eq{UCPSP}, with $H=H_{\bom}$,   $W=U$, and   a constant $\kappa$ independent of $\bom$ implies   a Wegner estimate.

 \begin{lemma} \label{lemWeg} Let $H_{\bom}$ be a crooked  Anderson Hamiltonian.   Let $I \subset ]-\infty,E_{0}]$
 be a closed interval and  $\Lambda=\La_L(x_0)$ a box centered at  $x_0\in \Z^d$ with $L\in \N_{\mathrm{odd}}$, $L\ge  2 + \delta_+ $. Suppose there exists a constant $\kappa>0$ such that
\beq \label{UCPI}
P_{\bom,\La }(I)U\up{\La}  P_{\bom,\La }(I)\ge  \kappa P_{\bom,\La }(I) \quad \text{with probability one}.
\eeq
Then 
\begin{align}\label{preWegner99}
\E \set{\tr P_{\bom,\La }(I) } \le  C_{d, \delta_+, V\up{0}_\infty}
   \pa{\kappa^{-2} (1 +E_0)}^{2^{1 + \frac {\log d}{\log 2}} }S_\La(\abs{I})\abs{\La}.
    \end{align}

 \end{lemma}

\begin{proof} We fix the box $\La$, let 
 $P= P_{\bom,\La }(I)$ for a closed interval $I \subset ]-\infty,E_{0}]$, and simply write $U$ for   $U\up{\La} $. Then it follows from \eq{UCPI}, using \eq{finsp}, that
\begin{align}\notag
\tr P &\le \kappa^{-1} \tr PUP = \kappa^{-1}\tr \sqrt{U}P  \sqrt{U} \le \kappa^{-2} \tr \sqrt{U}P UP \sqrt{U} 
 =  \kappa^{-2} \tr P UP U \\ \notag &=  \kappa^{-2} \tr P UP U P \le  \kappa^{-2} (1 +\alpha +E_0) \tr P U(H_{\bom,\La }+1+\alpha )^{-1} UP\\     \label{trp}
&\le  \kappa^{-2} (1+\alpha  +E_0) \tr P U(H_{0,\La }+1+\alpha )^{-1} UP\\ \notag
& =  \kappa^{-2} (1+\alpha  +E_0) \tr U P U(H_{0 ,\La}+1+\alpha )^{-1} \\ \notag
& =  \kappa^{-2} (1+\alpha  +E_0)  \sum_{i, j \in \widehat\La} \tr \sqrt{u_j} P \sqrt{u_i} T_{ij}, 
\end{align}
where
\beq
T_{ij}= \sqrt{u_i} (H_{0,\La }+1+\alpha)^{-1} \sqrt{u_j}  \qtx{for} i,j  \in \widehat\La.
\eeq

We now proceed as in \cite[Eqs.~(2.10)-(2.16)]{CHK2}, adapting \cite[Lemma A.1]{CHK2}.    Using  $\supp u_j \subset \La_{1 + \delta_+}(j)$, the resolvent identity (several times), trace estimates, and the Combes-Thomas estimate we obtain 
 \beq
 \norm{T_{ij}}_1 \le C_1 \e^{c_1 \abs{i-j}} \qtx{for all} i,j \in \widehat{\La} \qtx{with}  \abs{i-j}_\infty \ge  2 + \delta_+ , \label{Tdecay}
 \eeq
 where the constants $C_1$ and $c_1$ depend only on $d, \delta_+, V\up{0}_\infty$.
 Given $i \in \widehat{\La}$,  we set
  \beq
  \cJ _i =\set{j\in \widehat{\La};  \ ; \abs{i-j}_\infty <  2 + \delta_+}; \qtx{note that}  \# \cJ _i \le  \pa{2 + \delta_+}^d.
  \eeq
We have
\begin{align}\label{sumbreak}
& \sum_{i, j \in \widehat\La} \tr \sqrt{u_j} P \sqrt{u_i} T_{ij} = \sum_{i \in  \widehat\La} \set{\sum_{j\in \cJ_i^c} \tr \sqrt{u_j} P \sqrt{u_i} T_{ij} + \sum_{j\in \cJ_i} \tr \sqrt{u_j} P \sqrt{u_i} T_{ij}}.
 \end{align}
  Using spectral averaging \cite[Lemma~2.1]{CHK2} and  \eq{Tdecay}  we get
  \beq\label{spav}
 \E \abs{ \sum_{i \in \widehat{\La}  }\sum_{j\in \cJ_i^c} \tr \sqrt{u_j} P \sqrt{u_i} T_{ij}} \le C_2 S_\La(\abs{I})\abs{\La},
  \eeq
  where $C_2$  depends only on $d, \delta_+, V\up{0}_\infty$.
  
  Now let
  \beq
  T_\La= \sum_{i \in \widehat\La} \sum_{j\in \cJ_i}  \sqrt{u_i}T_{ij}\sqrt{u_j}= \sum_{i \in \La} \sum_{j\in \cJ_i}  {u_i}(H_{0,\La }+1+\alpha)^{-1}{u_j},
  \eeq
  so
  \beq\label{PTLa}
 \sum_{i \in \widehat\La} \sum_{j\in \cJ_i} \tr \sqrt{u_j} P \sqrt{u_i} T_{ij}=  \tr P  T_\La .
  \eeq
 Proceeding as in  in \cite[Eqs.~(A.4)-(A.5)]{CHK2}, we get
  \begin{align}
  \abs{\tr P  T_\La} \le 
  \pa{\sum_{j=1}^m    \tfrac {\sigma_j}{2^j \sigma_1\ldots\sigma_{j-1}}} \tr P + \tfrac {1}{2^m\sigma_1\ldots\sigma_{m}} \tr P \pa{T_\La T_\La^*}^{2^{m-1}},
 \end{align}
 for all $m\in \N$,  $\sigma_j >0$ for $j=1,2,\ldots,m$, and $\sigma_0=1$.
We take  $\beta=\pa{\kappa^{-2} (1 +E_0)}^{-1} $ and choose  $\sigma_j=  \beta^{2^{j-1}}$, so
 \begin{align}\label{PTLaalpha}
  \abs{\tr P  T_\La} \le 
 \beta\pa{1 - 2^{-m}} \tr P + 2^{-m} \beta^{1- 2^m} \tr P \pa{T_\La T_\La^*}^{2^{m-1}}.
 \end{align}
It follows from \eq{trp}, \eq{sumbreak}, \eq{spav}, \eq{PTLa}, \eq{PTLaalpha} that
\begin{align}
\E \tr P  &\le 
 C_2 \kappa^{-2} (1 +E_0+\alpha) S_\La(\abs{I})\abs{\La} +  \pa{1 - 2^{-m}}\E  \tr P \\
 & \hskip40pt  +2^{-m} \pa{\kappa^{-2} (1 +\alpha+E_0)}^{2^m}\E\set{ \tr P \pa{T_\La T_\La^*}^{2^{m-1}}},\notag
\end{align}
so
\begin{align}\label{soeq}
\E \tr P  & \le 
 C_22^m  \kappa^{-2} (1 +\alpha+E_0) S_\La(\abs{I})\abs{\La}\\
 \notag &  \hskip60pt +   \pa{\kappa^{-2} (1 +\alpha+E_0)}^{2^m}\E\set{ \tr P \pa{T_\La T_\La^*}^{2^{m-1}}}.\notag
\end{align}

We now estimate $\E\set{ \tr P \pa{T_\La T_\La^*}^{2^{m-1}}}$ as in \cite[Lemma~A.1]{CHK2}. Since we have 
$ {u_i}(H_{0,\La }+1+\alpha)^{-1}{u_j} \in \cT_{q}$ (i.e., $\tr \abs{ {u_i}(H_{0,\La }+1+\alpha)^{-1}{u_j}}^q<\infty$) for $q> \frac d 2$, letting 
\beq
m_d= \min \set{m \in \N; \quad  2^{m-1} >  \tfrac d 4  } = \min \set{m \in \N; \quad m >  \tfrac {\log d}{\log 2} -1 },
\eeq
we obtain, similarly to  \cite[Eq.~(A.8)]{CHK2}
\beq
\norm{ \pa{T_\La T_\La^*}^{2^{m_d-1}}}_1 \le  C_{d, \delta_+, V\up{0}_\infty} \abs{\La},
\eeq
and conclude, using spectral averaging as  in  \cite[Eqs.~(2.17)-(2.19)]{CHK2}, that 
\beq\label{soeq2}
\abs{\E\set{ \tr P \pa{T_\La T_\La^*}^{2^{m_d-1}}}} \le  C^\pr_{d, \delta_+, V\up{0}_\infty}S_\La(\abs{I}) \abs{\La}
\eeq

Putting together \eq{soeq} and \eq{soeq2} we get 
\begin{align}\label{soeq3}
\E \tr P  \le  C_{d, \delta_+, V\up{0}_\infty}
   \pa{\kappa^{-2} (1 +\alpha +E_0)}^{2^{m_d}} S_\La(\abs{I})\abs{\La},
    \end{align}
and \eq{preWegner99} follows, changing the constant to absorb $\alpha$ in case of periodic boundary condition.
\end{proof}

We are ready to prove Theorem~\ref{thmWegfull}.

\begin{proof}[Proof of Theorem~\ref{thmWegfull}]
Let $H_\bom$ be a  crooked Anderson Hamiltonian.
Given $E_{0}>0$, set $K=E_{0} + 2\pa{V\up{0}_\infty +MU_\infty} $, and define $\gamma$ by \eq{defgammaW}.   Given a box $\La$ as in  the theorem, we apply Theorem~\ref{thmUCPSP} with $H=H_0 +V\up{\La}_{\bom}$ and $W$ as in \eq{defWU}, concluding that for 
 any closed  interval $I \subset ]-\infty ,E_{0}]$ with $\abs{I} \le 2\gamma $  we have, using also \eq{Wprop22},
 \beq \label{chivchiWegfull}
 \Chi_{I}(H_{\bom,\La})\le  \gamma^{-2} \Chi_{I}(H_{\bom,\La}) W\up{\La} \Chi_{I}(H_{\bom,\La}) \le u_-^{-1}  \gamma^{-2}\Chi_{I}(H_{\bom,\La}) U\up{\La} \Chi_{I}(H_{\bom,\La}).
\eeq

We now apply  Lemma~\ref{lemWeg}, getting \eq{opWegner}.
\end{proof}

\section{At the bottom of the spectrum}\label{secbottom}

The following lemma is a slight extension of \cite[Theorem~1.1]{BLS}.

\begin{lemma}\label{lemBLS} Let  $H_{0}$ be a self-adjoint operator on a Hilbert space $\H$, bounded from below, and let $Y\ge 0$ be a bounded operator on $\H$.  Let $H(t)=H_{0}+ tY$ for $t\ge 0$, and set $E(t)= \inf \sigma (H(t))$, a non-decreasing function of $t$.  Let
$E(\infty)= \lim_{t\to \infty} E(t)=\sup_{t \ge 0} E(t)$.  Suppose $E(\infty) > E(0)$.  Given  $E_{1}\in ]E(0), E(\infty)[$,  let
\beq\label{kappa5}
\kappa = \kappa(H_0,Y,E_1)= \sup_{s >0; \ E(s) > E_{1}} \frac {E(s) -E_{1}}s   >0.
\eeq
Then  for all bounded operators $V \ge 0$ on $\H$ and Borel sets $B \subset ]-\infty, E_{1}]$ we have
\beq \label{PVP}
\Chi_{B}(H_{0}+V)Y\Chi_{B}(H_{0}+V) \ge \kappa \Chi_{B}(H_{0}+V).
\eeq
\end{lemma}

\begin{proof}
Fix    $E_{1}\in ]E(0), E(\infty)[$.  For all  Borel sets $B \subset ]-\infty, E_{1}]$ we have, writing $P_{V}(B)=\Chi_{B}(H_{0}+V)$,
\beq
P_{V}(B) (H_{0}+V)P_{V}(B) \le E_{1} P_{V}(B).
\eeq
Since $E_{1}\in ]E(0), E(\infty)[$, there is $s>0$ such that
$E(s) > E_{1}$. Then,
\begin{align}
P_{V}(B) (H(s)+V -sY -E_{1})P_{V}(B) = P_{V}(B) (H_{0}+V - E_{1})P_{V}(B)\le 0,
\end{align}
and hence, using $V \ge 0$, 
\begin{align}
 s P_{V}(B) Y P_{V}(B)&\ge P_{V}(B) (H(s)+V -E_{1})P_{V}(B)\\
&   \ge  P_{V}(B) (H(s)  -E_{1})P_{V}(B)\ge (E(s) -E_{1})P_{V}(B).\notag
\end{align}
The estimate \eq{PVP} follows 
\end{proof}

To use Lemma~\ref{lemBLS} we must show that $E(\infty) > E(0)$.  This will follow from the following lemma.

\begin{lemma}\label{lemEinfty}  Let $H_0$,  $u_-$, $W$ be as in   Definition~\ref{defAndHam}   and \eq{defWU},  set $H(t)= H_0 + t u_- W$ for $t\ge 0$, and let  
 $E(t)= \inf \sigma(H(t))$, $E(\infty)= \lim_{t\to \infty} E(t)=\sup_{t \ge 0} E(t)$.
   Then
   \beq\label{Dirichcons}
E(t)\ge  t u_- \delta_-^{ M_d \pa{1 + \pa{V\up{0}_\infty +2tu_- }^{\frac 2 3}}} \qtx{for all} t\ge 0,
\eeq
so we conclude that 
\beq\label{Dirichcons1}
E(\infty) \ge \sup_{t \in [0,\infty[} t  \delta_-^{ M_d \pa{1 + \pa{V\up{0}_\infty +2t   }^{\frac 2 3}}}>0.
\eeq
\end{lemma}

\begin{proof}  By our normalization  $E(0)=0$, and it follows from the min-max principle that  $0\le E(t_2)- E(t_1)\le (t_2- t_1)u_-$ for $0\le t_1 \le t_2$.  Thus $E(\infty) \in [0,\infty]$ is well defined.

Given a box  $\La=\La_L(x_0)$, where $x_0 \in \Z^d$ and $L\in \N_{\mathrm{odd}}$, $L \ge 72\sqrt{d}$,  set $E_{\La}\up{D}(t)=\inf \sigma(H_{\La}\up{D}(t))$. 
 Note that  $E_{\La}\up{D}(t)\ge E(t)\ge 0 $ for all $t \ge   0$ since we have Dirichlet boundary condition,
and we also  have    
\beq\label{infdelta}
E_{\La}\up{D}(t) \le \inf \sigma(- \Delta_\La\up{D})+ t u_-= d \pa{\tfrac{\pi}L}^2 + t u_-.
\eeq

Since $H_{\La}\up{D}(t)$ has compact resolvent,  there exists $\psi(t) \in \cD(\Delta_\La\up{D})$, $\norm{\psi(t)}=1$, such that $H_{\La}\up{D}(t)\psi(t) = E_{\La}\up{D}(t)\psi(t) $.  Applying Theorem~\ref{lemUCPeig} with $H=H_{\La}\up{D}(t)- E_{\La}\up{D}(t)$ and $\psi=\psi(t)$, and using \eq{defWU} and \eq{Wprop}, we get (see \cite[Proof of Theorem~4.9]{RV} for a similar argument)
\beq
 \delta_-^{M_d \pa{1 + \norm{V\up{0} +tu_- W -E_{\La}\up{D}(t) }_\infty^{\frac 2 3}}} \le \scal{\psi(t),W\psi(t)}.
\eeq
Using \eq{infdelta}, we get
\beq
\scal{\psi(t),W\psi(t)} \ge  \delta_-^{M_d \pa{1 + \pa{V\up{0}_\infty +2tu_- + d \pa{\tfrac{\pi}L}^2}^{\frac 2 3}}}\qtx{for all} t\ge 0.
\eeq
It follows that
\begin{align}
E_{\La}\up{D}(t) &\ge E_{\La}\up{D}(0)+ t u_- \delta_-^{M_d \pa{1 + \pa{V\up{0}_\infty +2tu_- + d \pa{\tfrac{\pi}L}^2}^{\frac 2 3}}}\\   \notag
&\ge  t u_- \delta_-^{M_d \pa{1 + \pa{V\up{0}_\infty +2tu_- + d \pa{\tfrac{\pi}L}^2}^{\frac 2 3}}}.
\end{align}
Taking $\La=\La_L(0)$  and noting that $\lim_{L\to \infty }E_{\La}\up{D}(t) =E(t) $, we get
\beq
E(t) \ge  t u_- \delta_-^{ M_d \pa{1 + \pa{V\up{0}_\infty +2tu_-  }^{\frac 2 3}}} \qtx{for all} t\ge 0,
\eeq
so we have \eq{Dirichcons}, and hence  \eq{Dirichcons1}, since
\beq\label{Dirichcons12}
E(\infty) \ge \sup_{t \in [0,\infty[} t u_-  \delta_-^{ M_d \pa{1 + \pa{V\up{0}_\infty +2tu_-   }^{\frac 2 3}}}=  \sup_{t \in [0,\infty[} t \delta_-^{ M_d \pa{1 + \pa{V\up{0}_\infty +2t  }^{\frac 2 3}}}.
\eeq
\end{proof}

We can now prove Theorem~\ref{thmWegbottom}.

\begin{proof}[Proof of Theorem~\ref{thmWegbottom}]
Let $H_\bom$ be a be a crooked Anderson Hamiltonian.  By Lemma~\ref{lemEinfty} we have $E(\infty)>0$, so we can pick $E_{1}\in ]0, E(\infty)[$,
and we have  \eq{kappabottom}.

Consider a box $\La=\La_L(x_0)$, where $x_0 \in \Z^d$ and $L\in \N_{\mathrm{odd}}$, $L \ge2 + \delta_+$. 
Using \eq{Dirichletbound}, we get  
\beq
 \kappa(H_{0,\La}\up{D},u_-W\up{\La},E_1)\ge \kappa= \kappa(H_{0},u_-W,E_1)>0,
  \eeq
 and  Lemma~\ref{lemBLS} then gives \eq{PVPb}. 
Applying  Lemma~\ref{lemWeg} we get  \eq{preWegnerdis}.
\end{proof}

We now turn to Theorem~\ref{thmloc}.

\begin{proof}[Proof of Theorem~\ref{thmloc}]  

 Let $H_{\bom,\lambda}$ be a crooked Anderson Hamiltonian with disorder $\lambda>0$, and assume $S(t) \le C t^\theta$, $\theta \in ]0,1]$.  By Theorem~\ref{thmWegbottom}, $E(\infty)>0$, so we  fix
 $E_{1}\in ]0, E(\infty)[$.   Let us pick $E_2 \in ]E_1, E(\infty)[$   and $t^*>0$  such that $E(t^*)\ge  E_2$.

Now  let $\La$ be a box as in    Theorem~\ref{thmWegbottom}.  Then
 \beq\label{bottomest}
 \P\set{H_{\bom,\lambda,\La}\up{D} \ge E_2  } \ge 1 -  \abs{\La} S\pa{\lambda^{-1}[0,t^*]}\ge 1 - C\pa{ \lambda^{-1}t^*}^{-\theta}  \abs{\La}.
 \eeq
 Moreover, we have the Wegner estimate \eq{preWegnerdis} (we omit  the dependence on parameters):
  \begin{align}\label{preWegnerdisbottom1111}
\E\set{ \tr P\up{D}_{\bom,\lambda,\La }(I) }  \le   C_{E_1}  
  \pa{\lambda^{-1} \abs{I}}^{-\theta}\abs{\La}.
    \end{align}  
   for any closed interval $I \subset  ]-\infty,E_{1}]$ and boxes $\La$ as  in Theorem~\ref{thmWegbottom}.
   
 Using \eq{bottomest} and \eq{preWegnerdisbottom1111}, we can prove Theorem~\ref{thmloc} by following the proof of  \cite[Theorem~3.1]{GKfinvol}.
 \end{proof}


\end{document}